\documentclass[11pt]{article}

\usepackage[utf8]{inputenc}
\usepackage[T1]{fontenc}
\usepackage{lmodern}

\usepackage{graphicx,comment}
\usepackage[linesnumbered,vlined]{algorithm2e}

\usepackage{amsmath,amssymb}
\usepackage{amsthm}
\usepackage{cite}
\usepackage{fullpage}

\newtheorem{theorem}{Theorem}[section]
\newtheorem{lemma}[theorem]{Lemma}

\newtheorem{claim}[theorem]{Claim}
\newtheorem{definition}[theorem]{Definition}

\newtheorem{fact}[theorem]{Fact}

\SetKwFor{Procedure}{procedure}{}{end procedure}

\DeclareMathOperator{\id}{id}

\newenvironment{lemma-repeat}[1]{\begin{trivlist}
		\item[\hspace{\labelsep}{\bf\noindent Lemma \ref{#1} }]\em }%
	{\end{trivlist}}
\newenvironment{theorem-repeat}[1]{\begin{trivlist}
		\item[\hspace{\labelsep}{\bf\noindent Theorem \ref{#1} }]\em }%
	{\end{trivlist}}

\newcommand*\samethanks[1][\value{footnote}]{\footnotemark[#1]}

\newcommand{\accept}{\texttt{accept}}
\newcommand{\reject}{\texttt{reject}}
\newcommand{\poly}{\textrm{poly}}

\SetKwComment{tcpy}{\#~}{} 
\SetKwFor{Perform}{perform}{times}{endp} 
\SetKwFor{Simul}{for each}{simultaneously}{ends} 
\SetKwFor{Use}{use}{to}{endu} 
\SetKwInput{KwVars}{Variables} 
\SetKwInput{KwSet}{Variables set} 

\RestyleAlgo{boxruled}
\LinesNumbered

\begin{document}
\begin{titlepage}
\title{Fast Distributed Algorithms for Testing Graph Properties}
\author{Keren Censor-Hillel\thanks{Technion -- Israel Institute of Technology, Department of Computer
  Science. \texttt{ckeren@cs.technion.ac.il},
  \texttt{eldar@cs.technion.ac.il}, \texttt{gregorys@cs.technion.ac.il},
  \texttt{yaduvasudev@gmail.com}. Supported in part by the Israel Science Foundation (grant 1696/14).}
	\and Eldar Fischer\samethanks 
	\and Gregory Schwartzman\samethanks 
	\and Yadu Vasudev\samethanks 
}

\maketitle
	
\begin{abstract}
%
We initiate a thorough study of \emph{distributed property testing} -- producing algorithms for the approximation problems of property testing in the CONGEST model. In particular, for the so-called \emph{dense} graph testing model we emulate sequential tests for nearly all graph properties having $1$-sided tests, while in the \emph{general} and \emph{sparse} models we obtain faster tests for triangle-freeness, cycle-freeness and bipartiteness, respectively. In addition, we show a logarithmic lower bound for testing bipartiteness and cycle-freeness, which holds even in the stronger LOCAL model.

In most cases, aided by parallelism, the distributed algorithms have a much shorter running time as compared to their counterparts from the sequential querying model of traditional property testing. The simplest property testing algorithms allow a relatively smooth transitioning to the distributed model. For the more complex tasks we develop new machinery that may be of independent interest.
\end{abstract}

~
\thispagestyle{empty}
\end{titlepage}

\newcommand{\ThmSim}
{
Any $\epsilon$-test in the dense graph model for a non-disjointed property that makes $q$ queries can be converted to a distributed $\epsilon$-test that takes $O(q^2)$ communication rounds.
}
\newcommand{\ThmTri}{
Algorithm~\ref{alg:triangle-freeness} is a distributed $\epsilon$-test in the general graph model for the property of containing no triangles, that requires $O(\epsilon^{-2})$ rounds.
}
\newcommand{\ThmBi}{
Algorithm~\ref{alg:dist-bip-test-det} is a distributed $\epsilon$-test in the bounded degree graph model for the property of being bipartite, that requires $O(\poly{(\epsilon^{-1} \log(n \epsilon^{-1}))})$ rounds.
}
\newcommand{\ThmCycle}{
Algorithm~\ref{alg:test-cycle-free} is a distributed $\epsilon$-test in the general graph model
for the property of being cycle-free, that requires $O(\log n/\epsilon)$ rounds.
}
\newcommand{\ThmLBBi}{
Any distributed $1/100$-test for the property of being bipartite requires
$\Omega(\log n)$ rounds of communication.
}
\newcommand{\ThmLBCycle}{
Any distributed $1/100$-test for the property of being cycle-free requires
$\Omega(\log n)$ rounds of communication.
}

\section{Introduction}
\label{sec:intro}

The performance of many distributed algorithms naturally depends on
properties of the underlying network graph. Therefore, an inherent goal is
to check whether the graph, or some given subgraph, has certain properties.
However, in some cases this is known to be hard, such as in the CONGEST
model~\cite{Peleg00}. In this model, computation proceeds in synchronous
rounds, in each of which every vertex can send an $O(\log{n})$-bit message
to each of its neighbors. Lower bounds for the number of rounds of type
$\tilde{\Omega}(\sqrt{n}+D)$ are known for \emph{verifying} many global graph
properties, where $n$ is the number of vertices in the network and $D$ is
its diameter (see, e.g. Das-Sarma et
al.~\cite{SarmaHKKNPPW12})\footnote{Here $\tilde{\Omega}$ hides factors that are
polylogarithmic in $n$.}.

To overcome such difficulties, we adopt the relaxation used in graph property testing, as first defined in \cite{GGR98,GoldreichR02}, to the distributed setting. That is, rather than aiming for an exact answer to the question of whether the graph $G$ satisfies a certain property $P$, we settle for distinguishing the case of satisfying $P$ from the case of being \emph{$\epsilon$-far} from it, for an appropriate measure of being far.

Apart from its theoretical interest, this relaxation is motivated by the common scenario of having distributed algorithms for some tasks that perform better given a certain property of the network topology, or given that the graph \emph{almost} satisfies that property. For example, Hirvonen et al.~\cite{HirvonenRSS14} show an algorithm for finding a large cut in triangle-free graphs (with additional constraints), and for finding an $(1-\epsilon)$-approximation if at most an $\epsilon$ fraction of all edges are part of a triangle. Similarly, Pettie and Su~\cite{PettieS15} provide fast algorithms for coloring triangle-free graphs.

We construct fast distributed algorithms for testing various graph properties. An important byproduct of this study is a toolbox that we believe will be useful in other settings as well.

\subsection{Our contributions}

We provide a rigorous study of property testing methods in the realm of distributed computing under the CONGEST model.
We construct \emph{$1$-sided error distributed $\epsilon$-tests}, in which if the graph satisfies the property then all vertices output \accept, and if it is $\epsilon$-far from satisfying the property then at least one vertex outputs \reject ~with probability at least $2/3$.
Using the standard amplification method of invoking such a test $O(\log{n})$ times and having a vertex output \reject~if there is at least one invocation in which it should output \reject, gives rejection with higher probability at the price of a multiplicative $O(\log{n})$ factor for the number of rounds.

The definition of a graph being $\epsilon$-far from satisfying a property is roughly one of the following (see Section~\ref{sec:prelim} for precise definitions): (1) Changing any $\epsilon n^2$ entries in the adjacency matrix does not give a graph that satisfies the property (dense model), or (2) changing any $\epsilon \cdot\max\{n,m\}$ entries in the adjacency matrix does not give a graph that satisfies the property, where $m$ is the number of edges (general model). A particular case here is when the degrees are bounded by some constant $d$, and any resulting graph must comply with this restriction as well (sparse model).

In a \emph{sequential $\epsilon$-test}, access to the input is provided by queries, whose type depends on the model. In the dense model these are asking whether two vertices $v,u$ are neighbors, and in the general and sparse models these can be either asking what the degree of a vertex $v$ is, or asking what the $i$-th neighbor of $v$ is (the ordering of neighbors is arbitrary). While a sequential $\epsilon$-test can touch only a small handful of vertices with its queries, in a distributed test the lack of ability to communicate over large distances is offset by having all $n$ vertices operating in parallel.


Our first contribution is a general scheme for a near-complete emulation in the distributed context of $\epsilon$-tests originating from the dense graph model (Section~\ref{sec:emulation}). This makes use of the fact that in the dense model all (sequential) testing algorithms can be made \emph{non-adaptive}, which roughly means that queries do not depend on responses to previous queries (see Section~\ref{sec:prelim} for definition).
In fact, such tests can be made to have a very simple structure, allowing the vertices in the distributed model to ``band together'' for an emulation of the test. There is only one additional technical condition (which we define below), since in the distributed model we cannot handle properties whose counter-examples can be ``split'' to disjoint graphs. For example, the distributed model cannot hope to handle the property of the graph having no disjoint union of two triangles, a property for which there exists a test in the dense model.

\begin{theorem-repeat}{thm:sim}
\ThmSim
\end{theorem-repeat}

We next move away from the dense graph model to the sparse and general models, that are sometimes considered to be more realistic. In the general model, there exists no test for the property of containing no triangle that makes a number of queries independent of the number of graph vertices~\cite{AKKR}. Here the distributed model can do better, because the reason for this deficiency is addressed by having all vertices operate concurrently. In Section~\ref{sec:triangle-freeness} we adapt the interim lemmas used in the best testing algorithm constructed in~\cite{AKKR}, and construct a distributed algorithm whose number of rounds is independent of $n$.

\begin{theorem-repeat}{thm:tri}
\ThmTri
\end{theorem-repeat}

The sparse and general models inherently require \emph{adaptive} property testing algorithms, since there is no other way to trace a path from a given vertex forward, or follow its neighborhood. Testing triangle freeness sequentially uses adaptivity only to a small degree.
However, other problems in the sparse and general models, such as the one we explore next, have a high degree of adaptivity built into their sequential algorithms, and we need to take special care for emulating it in the distributed setting.

In the sparse model (degrees bounded by a constant $d$), we adapt ideas from the bipartiteness testing algorithm of~\cite{GR99}, in which we search for odd-length cycles. Here again the performance of a distributed algorithm surpasses that of the test (a number of rounds polylogarithmic in $n$ vs. a number of queries which is $\Omega(\sqrt{n})$ -- a lower bound that is given in~\cite{GoldreichR02}). The following is proved in Section~\ref{sec:bi}.

\begin{theorem-repeat}{thm:bi}
\ThmBi
\end{theorem-repeat}

In the course of proving Theorem~\ref{thm:bi} we develop a method that we consider to be of independent interest\footnote{This technique was recently independently and concurrently devised in~\cite{GhaffariKS16} for a different use.}. The algorithm works by performing $2n$ random walks concurrently (two starting from each vertex). The parallel execution of random walks despite the congestion restriction is achieved by making sure that the walks have a uniform stationary distribution, and then showing that congestion is ``close to average'', which for the uniform stationary distribution is constant.

In Section~\ref{sec:cycle} we show a fast test for cycle-freeness. This makes use of a combinatorial lemma that we prove, about cycles that remain in the graph after removing edges independently with probability $\epsilon/2$. 
The following summarizes our result for testing cycle-freeness.
\begin{theorem-repeat}{thm:cycle-free-correctness}
\ThmCycle
\end{theorem-repeat}

We also prove lower bounds for testing bipartiteness and cycle-freeness (matching the upper bound for the latter). Roughly speaking, these are obtained by using the probabilistic method with alterations to construct graphs which are far from being bipartite or cycle-free, but all of their cycles are of length that is at least logarithmic. This technique bears some similarity to the classic result by Erd{\"o}s \cite{erdos1959graph}, which showed the existence of graphs with large girth and large chromatic number. The following are given in Section~\ref{sec:lowerbounds}.
\begin{theorem-repeat}{thm:bipart-lower-bound}
\ThmLBBi
\end{theorem-repeat}

\begin{theorem-repeat}{cor:cycle-lower-bound}
\ThmLBCycle
\end{theorem-repeat}
%
%
%
\paragraph{Roadmap:}
The paper is organized as follows. The remainder of this section consists of related work and historical background on property testing.
Section~\ref{sec:prelim} contains formal definitions and some mathematical tools.
The emulation of sequential tests for the dense model is given in Section~\ref{sec:emulation}. In Section~\ref{sec:triangle-freeness} we give our distributed test for triangle-freeness.
In Section~\ref{sec:bi} we provide a distributed test for bipartiteness, along with our new method of executing many random walks, and in Section~\ref{sec:cycle} we give our test for cycle-freeness. Section~\ref{sec:lowerbounds} gives our logarithmic lower bounds for testing bipartiteness and cycle-freeness.
We conclude with a short discussion in Section~\ref{sec:discussion}.

\subsection{Related work}


The only previous work that directly relates to our distributed setting is due to Brakerski and Patt-Shamir~\cite{brakerski2011distributed}. They show a \emph{tolerant} property testing algorithm for finding large (linear in size) \emph{near-cliques} in the graph. An $\epsilon$-near clique is a set of vertices for which all but an $\epsilon$-fraction of the pairs of vertices have an edge between them. The algorithm is tolerant, in the sense that it finds a linear near-clique if there exists a linear $\epsilon^3$-near clique. That is, the testing algorithm considers two thresholds of being close to having the property (in this case -- containing a linear size clique). We are unaware of any other work on property testing in this distributed setting.

Testing in a different distributed setting was considered in Arfaoui et al.~\cite{ArfaouiFIM14}. They study testing for cycle-freeness, in a setting where each vertex may collect information of its entire neighborhood up to some distance, and send a short string of bits to a central authority who then has to decide whether the graph is cycle-free or not.

Related to having information being sent to, or received by, a central authority, is the concept of proof-labelling schemes, introduced by Korman et al.~\cite{KormanKP10} (for extensions see, e.g., Baruch et al.~\cite{BaruchFP15}). In this setting, each vertex is given some external label, and by exchanging labels the vertices need to decide whether a given property of the graph holds. This is different from our setting in which no information other than vertex IDs is available. Another setting that is related to proof-labelling schemes, but differs from our model, is the prover-verifier model of Foerster et al.~\cite{FoersterLSW16}.

Sequential property testing has the goal of computing without processing the entire input. The wider family of \emph{local computation algorithms} (LCA) is known to have connections with distributed computing, as shown by Parnas and Ron~\cite{ParnasR07} and later used by others. A recent study by G\"{o}\"{o}s et al.~\cite{goos2015non} proves that under some conditions, the fact that a centralized algorithm can query distant vertices does not help with speeding up computation. However, they consider the LOCAL model, and their results apply to certain properties that are not influenced by distances.

Finding induced subgraphs is a crucial task and has been studied in several different distributed models (see, e.g.,~\cite{KariMRS15,DruckerKO12,Censor-HillelKK15,DolevLP12}). Notice that for \emph{finding} subgraphs, having \emph{many} instances of the desired subgraph can help speedup the computation, as in~\cite{DolevLP12}. This is in contrast to algorithms that perform faster if there are \emph{no} or only \emph{few} instances, as explained above, which is why we test for, e.g., the property of being \emph{triangle-free}, rather for the property of \emph{containing} triangles. (Notice that these are not the same, and in fact every graph with $3/\epsilon$ or more vertices is $\epsilon$-close to having a triangle.)

Parallelizing many random walks was addressed in~\cite{AlonAKKLT11}, where the question of graph covering via random walks is discussed. It is shown there that for certain families of graphs there is a substantial speedup in the time it takes for $k$ walks starting from the same vertex to cover the graph, as compared to a single walk. No edge congestion constraints are taken into account.
In~\cite{SarmaNPT13}, it is shown how to perform, under congestion, a single random walk of length $L$ in $\tilde{O}(\sqrt{LD})$ rounds, and $k$ random walks in $\tilde{O}(\sqrt{kLD} +k)$ rounds, where $D$ is the diameter of the graph. Our method has no dependence on the diameter, allowing us to perform a multitude of \emph{short walks} much faster.

\subsection{Historical overview}
The first papers to consider the question of property testing were~\cite{BLR93} and~\cite{RS96}.
The original motivations for defining property testing were its connection to some Computerized
Learning models, and the ability to leverage some properties to construct Probabilistically
Checkable Proofs (PCPs -- this is related to property testing through the areas of Locally Testable
Codes and Locally Decodable Codes, LTCs and LDCs). Other motivations since then have entered the
fray, and foremost among them are sublinear-time algorithms, and other big-data considerations.
Since virtually no property can be decidable without reading the entire input, property testing
introduces a notion of the allowable approximation to the original problem. In general, the
algorithm has to distinguish inputs satisfying the property, from inputs that are {\em
$\epsilon$-far} from it.  For more information on the general scheme of ``classical'' property
testing, consult the surveys \cite{Ron08,Fischer,GR10}.

The older of the graph testing models discussed here is the dense model, as defined in the seminal work of Goldreich, Goldwasser and Ron~\cite{GGR98}.
The dense graph model has historically kick-started combinatorial property testing in earnest, but it has some shortcomings. Its main one is the distance function, which makes sense only if we consider graphs having many edges (hence the name ``dense model'') -- any graph with $o(n^2)$ edges is indistinguishable in this model from an empty graph.

The stricter and at times more plausible distance function is one which is relative to the actual number of edges, rather than the maximum $\binom{n}{2}$. The general model was defined in~\cite{AKKR}, while the sparse model was defined already in~\cite{GoldreichR02}.
The main difference between the sparse and the general graph models is that in the former there is also a guaranteed upper bound $d$ on the degrees of the vertices, which is given to the algorithm in advance (the query complexity may then depend on $d$, either explicitly, or more commonly implicitly by considering $d$ to be a constant).

\section{Preliminaries}
\label{sec:prelim}
\subsection{Additional background on property testing}
While the introduction provided rough descriptions of the different property testing models, here we provide more formal definitions. The dense model for property testing is defined as follows.

\begin{definition}[dense graph model \cite{GGR98}]
	The dense graph model considers as objects graphs that are given by their adjacency matrix. Hence it is defined by the following features.
	\begin{itemize}
		\item {\bf Distance:} Two graphs with $n$ vertices each are considered to be {\em $\epsilon$-close} if one can be obtained from the other by deleting and inserting at most $\epsilon n^2$ edges (this is, up to a constant factor, the same as the normalized Hamming distance).
		\item {\bf Querying scheme:} A single query of the algorithm consists of asking whether two vertices $u,v\in V$ form a graph edge in $E$ or not.
		\item {\bf Allowable properties:} All properties have to be invariant under permutations of the input that pertain to graph isomorphisms (a prerequisite for them being graph properties).
	\end{itemize}
	The number of vertices $n$ is given to the algorithm in advance.
\end{definition}



As discussed earlier, the sparse and general models for property testing relate the distance function to the actual number of edges in the graph. They are formally defined as follows.

\begin{definition}[sparse \cite{GoldreichR02} and general \cite{AKKR} graph models]
	These two models consider as objects graphs given by their adjacency lists. They are defined by the following features.
	\begin{itemize}
		\item {\bf Distance:} Two graphs with $n$ vertices and $m$ edges (e.g.\ as defined by the denser of the two) are considered to be {\em $\epsilon$-close} if one can be obtained from the other by deleting and inserting at most $\epsilon\max\{n,m\}$ edges\footnote{Sometimes in the sparse graph model the allowed number of changes is $\epsilon dn$, as relates to the maximum possible number of edges; when $d$ is held constant the difference is not essential.}.
		\item {\bf Querying scheme:} A single query consists of either asking what is the degree of a vertex $v$, or asking what is the $i$'th neighbor of $v$ (the ordering of neighbors is arbitrary).
		\item {\bf Allowable properties:} All properties have to be invariant under graph isomorphisms (which here translate to a relabeling that affects both the vertex order and the neighbor ids obtained in neighbor queries), and reordering of the individual neighbor lists (as these orderings are considered arbitrary).
	\end{itemize}
\end{definition}

In this paper, we mainly refer to the distance functions of these models, and less so to the querying scheme, since the latter will be replaced by the processing scheme provided by the distributed computation model. Note that most property testing models get one bit in response to a query, e.g.,
``yes/no'' in response to ``is uv an edge'' in the dense graph model. However, the sparse and general models may receive $\log{n}$ bits of information for one query, e.g., an id of a neighbor of a vertex. Also, the degree of a vertex, which can be given as an answer to a query in the general model, takes $\log{n}$ bits.
Since the distributed CONGEST model allows passing a vertex id or a vertex degree along an edge in $O(1)$ rounds, we can equally relate to all three graph models.

Another important point is the difference between $1$-sided and $2$-sided testing algorithms, and the difference between non-adaptive and adaptive algorithms.

\begin{definition}[types of algorithms]
	A property testing algorithm is said to have {\em $1$-sided error} if there is no
	possibility of error on accepting satisfying inputs. That is, an input that satisfies the
	property will be accepted with probability $1$, while an input $\epsilon$-far from the
	property will be rejected with a probability that is high enough (traditionally this means a
	probability of at least $2/3$). A {\em $2$-sided  error} algorithm is also allowed to reject
	satisfying inputs, as long as the probability for a correct answer is high enough
	(traditionally at least $2/3$).
	
	A property testing algorithm is said to be {\em non-adaptive} if it decides all its queries in advance (i.e.\ based only on its internal coin tosses and before receiving the results of any query), while only its accept/reject output may depend on the actual input. An {\em adaptive} algorithm may make each query in turn based on the results of its previous queries (and, as before, possible internal coin tosses).
\end{definition}

In the following we address both adaptive and non-adaptive algorithms. However, we restrict ourselves to $1$-sided error algorithms, since the notion of $2$-sided error is not a good match for our distributed computation model.

\subsection{Mathematical background}

An important role in our analyses is played by the Multiplicative Chernoff Bound (see, e.g.,~\cite{Mitzenmacher}), hence we state it here for completeness.
\begin{fact}
\label{fact:chernoff}
Suppose that $X_1, ..., X_n$ are independent random variables taking values in $\{0, 1\}$. Let $X$ denote their sum and let $\mu = E[X]$ denote its expected value. Then, for any $\delta > 0$,
\begin{align*}
Pr[X < (1-\delta)\mu] < (\frac{e^{-\delta}}{(1-\delta)^{(1-\delta)}})^{\mu},\\
Pr[X > (1+\delta)\mu] < (\frac{e^{\delta}}{(1+\delta)^{(1+\delta)}})^{\mu}.
\end{align*}
Some convenient variations of the bounds above are:
\begin{align*}
Pr[X \geq (1+\delta)\mu] < e^{-\delta \mu / 3}, \quad \delta \geq 1\\
Pr[X \geq (1+\delta)\mu] < e^{-\delta^2 \mu / 3}, \quad \delta \in (0,1)\\
Pr[X \leq (1-\delta)\mu] < e^{-\delta^2 \mu / 2}, \quad \delta \in (0,1).
\end{align*}
\end{fact}

\section{Distributed emulation of sequential tests in the dense model}
\label{sec:emulation}
We begin by showing that under a certain assumption of being \emph{non-disjointed}, which we define below, a property $P$ that has a sequential test in the dense model that requires $q$ queries can be tested in the distributed setting within $O(q^2)$ rounds.
We prove this by constructing an emulation that translates sequential tests to distributed ones. For this we first introduce a definition of a \emph{witness} graph and then adapt~\cite[Theorem 2.2]{Goldreich}, restricted to $1$-sided error tests, to our terminology.

\begin{definition}
\label{def:rejected-graph}
Let $P$ be a property of graphs with $n$ vertices. Let $G'$ be a graph with $k<n$ vertices. We say that $G'$ is a \emph{witness against $P$}, if it is not an induced subgraph of any graph that satisfies $P$.		
\end{definition}
Notice that if $G'$ has an induced subgraph $H$ that is a witness against $P$, then by the above
definition $G'$ is also a witness against $P$.\par
The work of \cite{Goldreich} transforms tests of graphs in the dense model to a canonical form where the query scheme is based on vertex selection. This is useful in particular for the distributed model, where the computational work is essentially based in the vertices. We require the following special case for 1-sided error tests.

\begin{lemma}[\textbf{\cite[Theorem 2.2]{Goldreich}}]
\label{claim:canonical-tester}
Let $P$ be a property of graphs with $n$ vertices.
If there exists a $1$-sided error $\epsilon$-test for $P$ with query complexity $q(n, \epsilon)$, then there exists a $1$-sided error $\epsilon$-test for $P$ that uniformly selects a set of $q'=2q(n, \epsilon)$ vertices, and accepts if and only if the induced subgraph is not a witness against $P$.
\end{lemma}

Our emulation leverages Lemma~\ref{claim:canonical-tester} under an assumption on the property $P$, which we define as follows.
\begin{definition}
We say that $P$ is a \emph{non-disjointed} property if for every graph $G$ that does not satisfy $P$ and an induced subgraph $G'$ of $G$ such that $G'$ is a witness against $P$, $G'$ has some connected component which is also a witness against $P$. We call such components \emph{witness components}.
\end{definition}

We are now ready to formally state our main theorem for this section.

\begin{theorem}
\label{thm:sim}
\ThmSim
\end{theorem}

The following lemma essentially says that not satisfying a non-disjointed property cannot rely on subgraphs that are not connected, which is exactly what we need to forbid in a distributed setting.
\begin{lemma}
\label{lem:dist-prop-min}
The property $P$ is a non-disjointed property if and only if all minimal
witnesses that are induced subgraphs of $G$ are connected.
\end{lemma}
Here \emph{minimal} refers to the standard terminology, which means that no proper induced subgraph is a witness against $P$.

\begin{proof}
First, if $P$ is non-disjointed and $G$ does not satisfy $P$, then for
every subgraph $G'$ of $G$ that is a witness against $P$, $G'$ has a
witness component. If $G'$ is minimal then it must be
connected, since otherwise it contains a connected component which is a
witness against $P$, which contradicts the minimality of $G$.

For the other direction, if all the minimal witnesses that are induced
subgraphs of $G$ are connected, then every induced subgraph $G'$ that is a
witness against $P$ is either minimal, in which case it is connected, or is
not minimal, in which case there is a subgraph $H$ of $G'$ which is
connected and a minimal witness against $P$. The connected component $C$ of
$G'$ which contains $H$ is a witness against $P$ (otherwise $H$ is not a
witness against $P$), and hence it follows that $P$ is non-disjointed.
\end{proof}

Next, we give the distributed test (Algorithm~\ref{alg:simulate-test}). The test has an outer loop
in which each vertex picks itself with probability $5q/n$, collects its neighborhood of a certain
size of edges between \emph{picked} vertices in an inner loop, and rejects if it identifies a
witness against $P$. The outer loop repeats two times because not only does the sequential test have
an error probability, but also with some small probability we may randomly pick too many or not
enough vertices in order to emulate it. Repeating the main loop twice reduces the error probability
back to below $1/3$. In the inner loop, each vertex collects its neighborhood of picked vertices and
checks if its connected component is a witness against $P$. To limit communications this is done
only for components of picked vertices that are sufficiently small: if a vertex detects that it is
part of a component with too many edges then it accepts and does not participate until the next
iteration of the outer loop.

\begin{algorithm}[htbp]
\caption{Emulation algorithm with input $q$ for property $P$
\label{alg:simulate-test}}
\KwVars{$U_v$ edges known to $v$, $U'_v$ edges to update and send (temporary variables)}
\Perform{$2$}
{
	reset the state for all vertices\\
			
	\Simul{ vertex $v$  }
	{
		Vertex $v$ \textit{picks} itself with probability $5q/n$\\
		\If{$v$ is picked}
		{
			Notify all neighbors that $v$ is picked\\
			Set $U'_v = \{(v,u) \in E | \text{  u is picked}\}$ and $U_v=\emptyset$\\
					
			\Perform{$10q$}
			{
				\tcpy{At each iteration $U_v$ is a subgraph of $v$'s connected component}
				$U'_v = U'_v \backslash U_v$ \tcpy{only need recently discovered edges}
				$U_v = U_v\cup U'_v$ \tcpy{add them to $U_v$}
						
				\If(\tcpy*[h]{don't operate if there are too many edges}){$|U_v|\leq 100q^2$}
				{
					Send $U'_v$ to all picked neighbours of $v$ \tcpy{propagate known edges}
				}
				\textit{Wait} until the time bound for all other vertices to finish this iteration\\
				Set $U'_v$ to the union of edge sets received from neighbors
			}
			\If{ $U_v \cup U'_v$ is a witness against $P$}
			{
				Vertex $v$ outputs \reject ~(ending all operations)
			}
		}
		\Else
		{
			\textit{Wait} until the time bound for all other vertices to finish this iteration of the outermost loop
		}
	}
}
Every vertex $v$ that did not \reject~outputs \accept

\end{algorithm}

To analyze the algorithm, we begin by proving that there is a constant probability for the number of picked vertices to be sufficient and not too large.
\begin{lemma}
\label{lem:q}
The probability that the number of vertices picked by the algorithm is between $q$ and $10q$ is more than $2/3$ .
\end{lemma}
\begin{proof}
For every $v \in V$, we denote by $X_v$ the indicator variable for the event that vertex $v$ is picked.
Note that these are all independent random variables. Using the notation $X = \sum_{v\in V}X_v$ gives that $E[X]=5q$, because each vertex is picked with probability $5q/n$.
Using the Chernoff Bound from Fact~\ref{fact:chernoff} with $\delta=4/5$ and $\mu = 5q$, we can bound the probability of having too few picked vertices:
\begin{align*}
Pr[X < q] = Pr[X < (1-\delta)\mu] < (\frac{e^{-4/5}}{(1-(4/5))^{(1-(4/5))}})^{5q} = (\frac{5}{e^4})^q < \frac{1}{10}.
\end{align*}
For bounding the probability that there are too many picked vertices, we use the other direction of the Chernoff Bound with $\delta=1$ and $\mu = 5q$, giving:
\begin{align*}
Pr[X > 10q]  = Pr[X > (1+\delta)\mu] < (\frac{e}{2^2})^{5q} = (\frac{e^5}{2^{10}})^q < \frac{2}{10}.
\end{align*}
Thus, with probability at least $2/3$ it holds that $q \leq X \leq 10q$.
\end{proof}

Now, we can use the guarantees of the sequential test to obtain the guarantees of our algorithm.
\begin{lemma}
\label{lemma:correctness}
Let $P$ be a non-disjointed graph property. If $G$ satisfies $P$ then all vertices output \accept ~in Algorithm~\ref{alg:simulate-test}. If $G$ is $\epsilon$-far from satisfying $P$, then with probability at least $2/3$ there exists a vertex that outputs \reject.
\end{lemma}
	
\begin{proof}
First, assume that $G$ satisfies $P$. Vertex $v$ outputs \reject~only if it is part of a witness against $P$, which is, by definition, a component that cannot be extended to some $H$ that satisfies $P$. However, every component is an induced subgraph of $G$ itself, which does satisfy $P$, and thus every component can be extended to $G$. This implies that no vertex $v$ outputs \reject.

Now, assume that $G$ is $\epsilon$-far from satisfying $P$. Since the sequential test rejects with probability at least $2/3$, the probability that a sample of at least $q$ vertices induces a graph that cannot be extended to a graph that satisfies $P$ is at least $2/3$ . Because $P$ is non-disjointed, the induced subgraph must have a connected witness against $P$.
We note that a sample of more than $q$ vertices does not reduce the rejection probability. Hence, if we denote by $A$ the event that the subgraph induced by the picked vertices has a connected witness against $P$, then $Pr[A] \geq 2/3$, conditioned on that at least $q$ vertices were picked.

However, a sample that is too large may cause a vertex to output \accept ~because it cannot collect its neighborhood. We denote by $B$ the event that the number of vertices sampled is between $q$ and $10q$, and by Lemma~\ref{lem:q} its probability is at least $2/3$. We bound $Pr[A\cap B]$ using Bayes' Theorem, obtaining $Pr[A\cap B]= Pr[A|B]Pr[B] \geq (2/3)^2$. Since the outer loop consists of $2$ independent iterations, this gives a  probability of at least $1-(1-4/9)^2 \geq 2/3$ for having a vertex that outputs \reject.
\end{proof}

We now address the round complexity. Each vertex only sends and receives information from its $q$-neighborhood about edges between the chosen vertices. If too many vertices are chosen we detect this and accept. Otherwise we only communicate the chosen vertices and their edges, which requires $O(q^2)$ communication rounds using standard \emph{pipelining}\footnote{Pipelining means that each vertex has a buffer for each edge, which holds the information (edges between chosen vertices, in our case) it needs to send over that edge. The vertex sends the pieces of information one after the other.}.
Together with Lemma~\ref{lemma:correctness}, this proves Theorem~\ref{thm:sim}.
	
\subsection{Applications: $k$-colorability and perfect graphs}
\label{subsec:sim}
Next, we provide some examples of usage of Theorem~\ref{thm:sim}. A result by Alon and Shapira~\cite{AlonShapira} states that all graph properties closed under induced subgraphs are testable in a number of queries that depends only on $\epsilon^{-1}$. We note that, except for certain specific properties for which there are ad-hoc proofs, the dependence is usually a tower function in $\epsilon^{-1}$ or worse (asymptotically larger).


 From this, together with Lemma~\ref{claim:canonical-tester} and Theorem~\ref{thm:sim}, we deduce that if $P$ is a non-disjointed property closed under induced subgraphs, then it is testable, for every fixed $\epsilon$, in a constant number of communication rounds.

\paragraph{Example -- $k$-colorability:} The property of being $k$-colorable is testable in a distributed manner by our algorithm. All minimal graphs that are witnesses against $P$ (not $k$-colorable) are connected, and therefore according to Lemma~\ref{lem:dist-prop-min} it is a non-disjointed property. It is closed under induced subgraphs, and by \cite{AlonK02} there exists a $1$-sided error $\epsilon$-test for $k$-colorability that uniformly picks $O(k\log(k)/\epsilon^2)$ vertices, and its number of queries is the square of this expression (note that the polynomial dependency was already known by~\cite{GGR98}). Our emulation implies a distributed $1$-sided error $\epsilon$-test for $k$-colorability that requires $O(\poly{(k\epsilon^{-1})})$ rounds.

\paragraph{Example -- perfect graphs:} A graph $G$ is said to be \emph{perfect} if for every induced subgraph $G'$ of $G$, the chromatic number of $G'$ equals the size of the largest clique in $G'$.
Another characterization of a perfect graph is via \emph{forbidden subgraphs}: a graph is perfect if and only if it does not have odd holes (induced cycles of odd length at least $5$) or odd anti-holes (the complement graph of an odd hole)~\cite{strong}. Both odd holes and odd anti-holes are connected graphs. Since these are all minimal witnesses against the property, according to Lemma~\ref{lem:dist-prop-min} it is a non-disjointed property.
Using the result of Alon-Shapira~\cite{AlonShapira} we know that the property of a graph being perfect is testable. Our emulation implies a distributed $1$-sided error $\epsilon$-test for being a perfect graph that requires a number of rounds that depends only on $\epsilon$.

\section{Distributed test for triangle-freeness}
\label{sec:triangle-freeness}
In this section we show a distributed $\epsilon$-test for triangle-freeness. Notice that since triangle-freeness is a non-disjointed property, Theorem~\ref{thm:sim} gives a distributed $\epsilon$-test for triangle-freeness under the dense model with a number of rounds that is $O(q^2)$, where $q$ is the number of queries required for a sequential $\epsilon$-test for triangle-freeness. However, for triangle-freeness, the known number of queries is a tower function in $\log(1/\epsilon)$~\cite{Fox2010}.

Here we leverage the inherent parallelism that we can obtain when checking the neighbors of a vertex, and show a test for triangle-freeness that requires only $O(\epsilon^{-2})$ rounds (Algorithm~\ref{alg:triangle-freeness}). Importantly, our algorithm works not only for the dense graph model, but for the general graph model (where distances are relative to the actual number of edges), which subsumes it. In the sequential setting, a test for triangle-freeness in the general model requires a number of queries that is some constant power of $n$ by \cite{AKKR}. Our proof actually follows the groundwork laid in \cite{AKKR} for the general graph model -- their algorithm picks a vertex and checks two of its neighbors for being connected, while we perform the check for all vertices in parallel.

\begin{algorithm}[htbp]
\caption{Triangle freeness test\label{alg:triangle-freeness}}
\Simul{vertex $v$ }
{
	\Perform{$32\epsilon^{-2}$}
	{	
		Pick $w_1, w_2 \in N(v), w_1 \neq w_2$ uniformly at random \\
		Send $w_2$ to $w_1$ \tcpy{Ask $w_1$ if it is a neighbor of $w_2$}
		\ForEach(\tcpy*[h]{Asked by $u$ if $v$ is a neighbor of $w$}){$w_u$ sent by $u \in N(v)$}
		{
			\If{$w_u \in N(v)$}
			{
				Send ``yes'' to $u$\\
			}
			\Else
			{
				Send ``no'' to $u$\\
			}
		}
		\If{received ``yes'' from $w_1$}
		{
					\reject ~(ending all operations)
		}
	}
}
\accept ~(for vertices that did not reject)
\end{algorithm}

\begin{theorem}
\label{thm:tri}
\ThmTri
\end{theorem}

Our line of proof follows that of~\cite{AKKR}, by distinguishing edges that connect two high-degree vertices from those that do not. Formally, let $b = 2\sqrt{\epsilon^{-1} m}$, where $m$ is the number of edges in the graph, and denote $B=\{v \in V \mid deg(v) \geq b\}$. We say that an edge $e=(u,v)$ is \emph{light} if $v \not\in B$ or $u \not\in B$, and otherwise, we say that it is \emph{heavy}. That is, the set of heavy edges is $H = \{(u,v)\in E \mid u \in B, v \in B\}$. We begin with the following simple claim about the number of heavy edges.

\begin{claim}
\label{claim:tri-heavy}
The number of heavy edges, $|H|$, is at most $\epsilon m / 2$.
\end{claim}

\begin{proof}
The number of heavy edges is $|H| \leq |B|(|B|-1)/2 < |B|^2/2$. Since $|B|b  \leq 2m$, we get that $|B| \leq \frac{2m}{b} = \frac{2m}{2\sqrt{\epsilon^{-1} m}} = \sqrt{\epsilon m}$.
This gives that $|H| \leq \frac{1}{2}|B|^2 \leq  \epsilon m / 2$.
\end{proof}

Next, we fix an iteration $i$ of the algorithm. Every vertex $v$ chooses two neighbors $w_1, w_2$. Let $A=\{(v,w_1)\in E \mid v \in V\setminus B\}$, where $w_1$ is the first of the two vertices chosen by the low-degree vertex $v$. Let $T=\{e \in E \mid \text{$e$ is a light edge in a triangle}\}$, and let $A_T = T\cap A$.
We say that an edge $(v,w_1) \in A_T$ is \emph{matched} if $(v,w_2)$ is in the same triangle as $(v,w_1)$.
If $(v,w_1) \in A_T$ is matched then $\{v,w_1,w_2\}$ is a triangle that is detected by $v$.

We begin with the following lemma that states that if $G$ is $\epsilon$-far from being triangle-free, then in any iteration $i$ we can bound the expected number of matched edges from below by $\epsilon^2/8$. Let $Y$ be the number of matched edges.

\begin{lemma}
\label{lem:matched}
The expected number of matched edges by a single iteration of the algorithm, $E[Y]$, is greater than $\epsilon^2 /8$.
\end{lemma}
\begin{proof}
For every $e \in A_T$, let $Y_e$ be a random variable indicating whether $e$ is matched. Then $Y=\sum_{e \in A_T} Y_e$, giving the following bound:
\begin{align}
\label{eq:Y}
E[Y| A_T] = E[\sum_{e \in A_T} Y_e | A_T] = 
\sum_{e \in A_T} Pr[\text{e is matched}] \geq |A_T| / b,
\end{align}
where the last inequality follows because a light edge in $A_T$ is chosen by a vertex with degree at most $b$, hence the third triangle vertex gets picked with probability at least $1/b$.

Next, we argue that $E[|A_T|] \geq |T|/b$. To see why, for every edge $e$, let $X_e$ be a random variable indicating whether $e \in A$.
Let $X=\sum_{e \in T}X_e=|A_T|$. Then,
\begin{align}
\label{eq:A_T}
E[|A_T|] = E[X] = E[\sum_{e \in T} X_e] = \sum_{e \in T} E[X_e] = \sum_{e \in T} Pr[e \in A] \geq |T|/b,
\end{align}
where the last inequality follows because a light edge has at least one endpoint with degree at most $b$. Hence, this edge gets picked by it with probability at least $1/b$.

It remains to bound $|T|$ from below, for which we claim that  $|T| \geq \epsilon m / 2$. To prove this, first notice that, since $G$ is $\epsilon$-far from being triangle free, it has at least $\epsilon m$ triangle edges, since otherwise we can just remove all of them and make the graph triangle free with less than $\epsilon m$ edge changes.
By Claim~\ref{claim:tri-heavy}, the number of heavy edges satisfies $|H| \leq \epsilon/2 m$. Subtracting this from the number of triangle edges gives that at least $\epsilon m /2$ edges are light triangle edges, i.e.,

\begin{align}
\label{eq:T}
|T| \geq \epsilon m /2.
\end{align}

Finally, by Inequalities (\ref{eq:Y}),  (\ref{eq:A_T}) and (\ref{eq:T}), using iterated expectation we get:
\begin{align*}
E[Y]=E_{A_T}[E[Y|A_T]]\geq E[\frac{|A_T|}{b}]\geq \frac{|T|}{b^2} \geq \dfrac{\epsilon m}{2} \frac{1}{4\epsilon^{-1} m} = \epsilon^2/8.
\end{align*}
\end{proof}
	
We can now prove the correctness of our algorithm, as follows.	

\begin{lemma}
\label{lem:tri-test}
If $G$ is triangle-free then all vertices output \accept ~in Algorithm~\ref{alg:triangle-freeness}. If $G$ is $\epsilon$-far from being triangle-free, then with probability at least 2/3 there exists a vertex that outputs \reject.
\end{lemma}

\begin{proof}
If $G$ is triangle free then in each iteration $v$ receives ``no'' from $w_1$ and after all iterations it returns \accept.

Assume that $G$ is $\epsilon$-far from being triangle-free. Let $Z_{i,v}$ be an indicator variable for the event that vertex $v$ detects a triangle at iteration $i$.
First, we note that the indicators are independent, since a vertex detecting a triangle does not affect the chance of another vertex detecting a triangle (note that the graph is fixed), and the iterations are done independently.
Now, let $Z=\sum_{i=1}^{32\epsilon^{-2}} \sum_{v \in V} Z_{i,v}$, and notice that $Z$ is the total number of detections over all iterations. Lemma~\ref{lem:matched} implies that for a fixed $i$, it holds that $E[\sum_{v \in V} Z_{i,v}] = E[Y] \geq \epsilon^2/8$, which sums to:
\begin{align*}
E[Z] = E\left[\sum_{i=1}^{32\epsilon^{-2}} \sum_{v \in V}  Z_{i,v}\right] = \sum_{i=1}^{32\epsilon^{-2}} E\left[\sum_v Z_{i,v}\right] \geq \sum_{i=1}^{32\epsilon^{-2}} \epsilon^2/8 = 4.
\end{align*}
Using the Chernoff Bound from Fact~\ref{fact:chernoff} with $\delta=3/4$ and $\mu \geq 4$ gives
\begin{align*}
Pr[Z < 1] \leq Pr[Z < (1-\delta)\mu] < (\frac{e^{-3/4}}{(1-(3/4))^{(1-(3/4))}})^{4} = 4/e^3 < 2/3,
\end{align*}
and hence with probability at least $2/3$ at least one triangle is detected and the associated vertex outputs \reject, which completes the proof.
\end{proof}

In every iteration, each vertex initiates only two messages of size $O(\log{n})$ bits, one sent to
$w_1$ and one sent back by $w_1$. Since there are $O(\epsilon^{-2})$ iterations, this implies that the
number of rounds is $O(\epsilon^{-2})$ as well. This, together with Lemma~\ref{lem:tri-test},
completes the proof of Theorem~\ref{thm:tri}.

\section{Distributed bipartiteness test for bounded degree graphs}
\label{sec:bi}
In this section we show a distributed $\epsilon$-test for being bipartite for graphs with degrees bounded by $d$.
Our test builds upon the sequential test of~\cite{GR99} and, as in the case of triangle freeness, takes advantage of the ability to parallelize queries. While the number of queries of the sequential test is $\Omega(\sqrt{n})$~\cite{GoldreichR02}, the number of rounds in the distributed test is only \emph{polylogarithmic} in $n$ and polynomial in $\epsilon^{-1}$. As in~\cite{GR99}, we assume that $d$ is a constant, and omit it from our expressions (it is implicit in the $O$ notation for $L$ below).

Let us first outline the algorithm of~\cite{GR99}, since our distributed test borrows from its framework and our analysis is in part derived from it. The sequential test basically tries to detect odd cycles. It consists of $T$ iterations, in each of which a vertex $s$ is selected uniformly at random and $K$ random walks of length $L$ are performed starting from the source $s$. If, in any iteration with a chosen source $s$, there is a vertex $v$ which is reached by an even prefix of a random walk and an odd prefix of a random walk (possibly the same walk), then the algorithm rejects, as this indicates the existence of an odd cycle. Otherwise, the algorithm accepts.
To obtain an $\epsilon$-test the parameters are chosen to be $T=O(\epsilon^{-1})$, $K=O(\epsilon^{-4}\sqrt{n}\log^{1/2}{(n\epsilon^{-1})})$, and $L=O(\epsilon^{-8}\log^6{n})$.

The main approach of our distributed test is similar, except that a key ingredient is that we can afford to perform much fewer random walks from every vertex, namely $O(\poly{(\epsilon^{-1}\log{n\epsilon^{-1}})})$. This is because we can run random walks in parallel originating from all vertices at once. However, a crucial challenge that we need to address is that several random walks may collide on an edge, violating its congestion bound. To address this issue, our central observation is that \emph{lazy} random walks (chosen to have a uniform stationary distribution) provide for a very low probability of having too many of these collisions at once. The main part of the analysis is in showing that with high probability there will never be too many walks concurrently in the same vertex, so we can comply with the congestion bound. We begin by formally defining the lazy random walks that we use.

\begin{definition}
\label{def:lazy-rw}
A {\em lazy} random walk over a graph $G$ with degree bound $d$ is a random walk, that is, a (memory-less) sequence of random variables $Y_1,Y_2,\ldots$ taking values from the vertex set $V$, where the transition probability $Pr[Y_k=v|Y_{k-1}=u]$ is $1/2d$ if $uv$ is an edge of $G$, $1-deg(u)/2d$ if $u=v$, and $0$ in all other cases.
\end{definition}
The stationary distribution for the lazy random walk of Definition~\ref{def:lazy-rw} is uniform~\cite[Section 8]{ron2010algorithmic}.
Next, we describe a procedure to handle one iteration of moving the random walks (Algorithm~\ref{alg:dist-move-walks}), followed by our distributed test for bipartiteness using lazy random walks from every vertex concurrently (Algorithm~\ref{alg:dist-bip-test-det}).
		
\begin{algorithm}[htbp]
\caption{Move random walks once with input $\xi$\label{alg:dist-move-walks}}
\KwVars{$W_v$ walks residing in $v$ (multiset), $H_v$ history of walks through $v$}
\KwIn{$\xi$, the maximum congestion per vertex allowed}
\tcpy{each walk is characterized by $(i,u)$ where $i$ is the number of actual moves and $u$ is the origin vertex}
\Simul{vertex $v$}
{
    \If(\tcpy*[h]{give up if exceeded the maximum allowed}){$|W_v|\leq\xi$}
    {
        \For{every $(i,u)$ in $W_v$}
        {
            draw next destination $w$ (according to the lazy walk scheme)\\
			\If(\tcpy*[h]{walk exits $v$}){$w\neq v$}
            {
				send $(i+1,u)$ to $w$\\
				remove $(i,u)$ from $W_v$
			}
		}
	}
    {\em wait} until the maximum time for all other vertices to process up to $\xi$ walks\\
	add the walks received by $v$ to $W_v$ and $H_v$ \tcpy{walks entering $v$}
}
\end{algorithm}

It is quite immediate that Algorithm \ref{alg:dist-move-walks} takes $O(\xi)$ communication rounds.
		
\begin{algorithm}[htbp]
\caption{Distributed bipartiteness test\label{alg:dist-bip-test-det}}
\KwVars{$W_v$ walks residing in $v$ (multiset), $H_v$ history of walks through $v$}
\Perform{$\eta= O(\epsilon^{-9} \log(n \epsilon^{-1}))$}
{
    \Simul{vertex $v$}
    {
        initialize $H_v$ and $W_v$ with two copies of the walk $(0,v)$
    }
    \Perform{$L=O(\epsilon^{-8}\log^6{n})$}
    {
        move walks using Algorithm \ref{alg:dist-move-walks} with input $\xi=\gamma+2=3(2\ln{n}+\ln{L})+2$
    }
    \Simul{vertex $v$}
    {
        \If{$H_v$ contains $(i,u)$ and $(j,u)$ for some $u$, even $i$ and odd $j$}
        {
            \reject ~(ending all operations) \tcpy{odd cycle found}
        }
    }
}
\accept ~(for vertices that did not reject)
\end{algorithm}

Our main result here is that Algorithm~\ref{alg:dist-bip-test-det} is indeed a distributed $\epsilon$-test for bipartiteness.

\begin{theorem}
\label{thm:bi}
\ThmBi
\end{theorem}

The number of communication rounds is immediate from the algorithm -- it is dominated by the $L$ calls to Algorithm \ref{alg:dist-move-walks}, making a total of $O(\xi L)$ rounds, which is indeed $O(\poly{(\epsilon^{-1} \log(n \epsilon^{-1}))})$. To prove the rest of Theorem~\ref{thm:bi} we need some notation, and a lemma from~\cite{GR99} that bounds from below the probabilities for detecting odd cycles if $G$ is $\epsilon$-far from being bipartite.

Given a source $s$, if there is a vertex $v$ which is reached by an even prefix of a random walk $w_i$ from $s$ and an odd prefix of a random walk $w_j$ from $s$, we say that walks $w_i$ and $w_j$ \emph{detect a violation}. Let $p_s(k,\ell)$ be the probability that, out of $k$ random walks of length $\ell$ starting from $s$, there are two that detect a violation.
Using this notation, $p_s(K,L)$ is the probability that the sequential algorithm outlined in the beginning rejects in an iteration in which $s$ is chosen. Since we are only interested in walks of length $L$, we denote $p_s(k)=p_s(k,L)$. A good vertex is a vertex for which this probability is bounded as follows.
\begin{definition}
\label{def:good}
A vertex $s$ is called \emph{good} if $p_s(K) \geq 1/10$.
\end{definition}

In \cite{GR99} it was proved that being far from bipartite implies having many good vertices.

\begin{lemma}[\cite{GR99}]
\label{lemma:GR}
If $G$ is $\epsilon$-far from being bipartite then at least an $\epsilon/16$-fraction of the vertices are good.
\end{lemma}

In contrast to~\cite{GR99}, we do not perform $K$ random walks from every vertex in each iteration, but rather only $2$. Hence, what we need for our analysis is a bound on $p_s(2)$. To this end, we use $K$ as a parameter, and express $p_s(2)$ in terms of $K$ and $p_s(K)$.

\begin{lemma}
\label{lemma:ps2}
For every vertex $s$, $p_s(2) \geq 2p_s(K)/K(K-1)$.
\end{lemma}
\begin{proof}
Fix a source vertex $s$. For every $i,j \in [K]$, let $q_{i,j}$ be the probability of walks $w_i,w_j$ from $s$ detecting a violation. Because different walks are independent, we conclude that for every $i \neq j$ it holds that $q_{i,j}=p_s(2)$. Let $A_{i,j}$ be the event of walks $w_i,w_j$ detecting a violation. We have
\begin{align*}
p_s(K) = Pr[\cup_{i,j}A_{i,j}] \leq \sum_{i,j} Pr[A_{i,j}] = p_s(2)K(K-1)/2,
\end{align*}
which implies that $p_s(2) \geq 2p_s(K)/K(K-1)$.
\end{proof}

Using this relationship between $p_s(2)$ and $K$ and $p_s(K)$, we prove that our algorithm is an
$\epsilon$-test. First we prove this for the random walks themselves, ignoring the possibility that
Algorithm \ref{alg:dist-move-walks} will skip moving random walks due to its condition in Line~$2$.

\begin{lemma}\label{lem:walks-detect}
If $G$ is $\epsilon$-far from being bipartite, and we perform $\eta$ iterations of starting $2$ random walks of length $L$ from every vertex, then the probability that no violation is detected is bounded by $1/4$.
\end{lemma}

\begin{proof}
Assume that $G$ is $\epsilon$-far from being bipartite. By Lemma~\ref{lemma:GR}, at least $n\epsilon/16$ vertices are good, which means that for each of these vertices $s$, $p_s(K) \geq 1/10$. This implies that $\sum_{s \in V} p_s(K) \geq n \epsilon/160$.
Now, let $X_{i,s}$ be a random variable indicating whether there are two random walks starting at $s$ that detect a violation. Let $X=\sum_{i=0}^{\eta}\sum_{s \in V}X_{i,s}$. We prove that $Pr[X<1] < 1/4$. First, we bound $E[\sum_{s \in V}X_{i,s}]$ for some fixed $i$:
\begin{eqnarray*}
E[X] &=& E\left[\sum_{i=0}^{\eta}\sum_{s \in V}X_{i,s}\right] = \sum_{i=0}^{\eta}\sum_{s \in V}E[X_{i,s}] \\
&=& \sum_{i=0}^{\eta}\sum_{s \in V}{p_s(2)} \geq \sum_{i=0}^{\eta}\sum_{s \in V}{\frac{2p_s(K)}{K(K-1)}} \\
&=& {\frac{2}{K(K-1)}}\sum_{i=0}^{\eta}\sum_{s \in V}{p_s(K)} \geq {\frac{2}{K(K-1)}}\sum_{i=0}^{\eta}\frac{n\epsilon}{160}\\
&=& \frac{\eta n \epsilon}{80K(K-1)} \geq \frac{\eta n \epsilon}{80K^2}.
\end{eqnarray*}
For $\eta = 320K^2/ n\epsilon = O(\epsilon^{-9} \log(n \epsilon^{-1}))$ it holds that $E[X] \geq 4$.
Using the Chernoff Bound of Fact~\ref{fact:chernoff} with $\delta=3/4$ and $\mu \geq 4$ gives:
\begin{align*}
Pr[X < 1] \leq Pr[X < (1-\delta)\mu] < (\frac{e^{-3/4}}{(1-(3/4))^{(1-(3/4))}})^{4} = \frac{4}{e^3} < 1/4,
\end{align*}
which completes the proof.
\end{proof}

As explained earlier, the main hurdle on the road to prove Theorem~\ref{thm:bi} is in proving that the allowed congestion will not be exceeded. We prove the following general claim about the probability for $k$ lazy random walks of length $\ell$ from each vertex to exceed a maximum \emph{congestion factor} of $\xi$ walks allowed in each vertex at the beginning of each iteration. Here, an iteration is a sequence of rounds in which all walks are advanced by one step (whether or not they actually switch vertices).
\begin{lemma}
\label{lemma:cong}
With probability at least $1-1/n$, running $k$ lazy random walks of length $\ell$ originating from
every vertex will not exceed the maximum \emph{congestion factor} of $\xi=\gamma+k=3(2\ln n + \ln
\ell)+k$ walks allowed in each vertex at the beginning of each iteration, if $\gamma > k$.
\end{lemma}
We show below that plugging $k=2$, $\ell=L$ and $\gamma=3(2\ln{n}+\ln{L})$ in Lemma~\ref{lemma:cong}, together with Lemma~\ref{lem:walks-detect}, gives the correctness of Algorithm \ref{alg:dist-bip-test-det}.

To prove Lemma~\ref{lemma:cong}, we argue that it is unlikely for any vertex to have more than $k+\gamma$ walks in any iteration. Given that this is indeed the case in every iteration, the lemma follows by a union bound. We denote by $X_{v,i}$ the random variable whose value is the number of random walks at vertex $v$ at the beginning of the $i$-th iteration. That is, it is equal to the size of the set $W_v$ in the description of the algorithm.

\begin{lemma}
\label{lem:exp-walks-in-vertex}
For every vertex $v\in V$ and every iteration $i$ it holds that $E[X_{v,i}] = k$.
\end{lemma}
\begin{proof}
Let us first define random variables for our walks.
Enumerating our $kn$ walks ($k$ from each of the $n$ vertices) arbitrarily, let $Y_1^r,Y_2^r,\ldots$ denote the sequence corresponding to the $r$'th walk, that is, $Y_i^r$ is the vertex where the $r$'th walk is stationed at the beginning of the $i$'th iteration. In particular, $X_{v,i}=|\{r:Y_i^r=v\}|$.

Now let us define new random variables $Z_i^t$ in the following manner: First, we choose uniformly at random a permutation $\sigma:[rk]\to [rk]$. Then we set $Z_i^t=Y_i^{\sigma(t)}$ for all $1\leq i\leq\ell$ and $1\leq t\leq kn$. The main thing to note is that for any fixed $t$, $Z_1^t,Z_2^t,\ldots$ is a random walk (as it is equal to one of the random walks $Y_1^r,Y_2^r,\ldots$). But also, for every $t$, $Z_1^t$ is uniformly distributed over the vertex set of $G$, because we started with exactly $k$ random walks from every vertex. Additionally, since the uniform distribution is stationary for our lazy walks, this means that the unconditional distribution of each $Z_i^t$ is also uniform.

Now, since $\sigma$ is a permutation, it holds that $X_{v,i}=|\{r:Y_i^r=v\}|=|\{t:Y_i^{\sigma(t)}=v\}|=|\{t:Z_i^t=v\}|$. The expectation (by linearity of expectation) is thus $E[X_{v,i}]=\sum_{t=1}^{kn}Pr[Z_i^t=v]=k$.
\end{proof}

We can now prove Lemma~\ref{lemma:cong}.
\begin{proof}[Proof of Lemma~\ref{lemma:cong}]
We first claim that for every iteration $i \in [\ell]$ and every vertex $v\in V$, with probability at least $1-1/\ell n$ it holds that $X_{v,i} \leq k+\gamma$.
To show this, first fix some $v \in V$.
Let $Z_{j,i}$ be the indicator variable for the event of walk $w_j$ residing at vertex $v$ at the beginning of iteration $i$, where $j\in[k n]$. Then $X_{v,i} = \sum_{j=1}^{k n}Z_{j,i}$, and the variables $Z_{j,i}$, where $j\in[k n]$, are all independent. We use the Chernoff Bound of Fact~\ref{fact:chernoff} with $\delta = \gamma / k \geq 1$ and $\mu = k$ as proven in Lemma~\ref{lem:exp-walks-in-vertex}, obtaining:
\begin{align*}
Pr[X_{v,i} > k + \gamma] = Pr[X_{v,i} > (\gamma / k + 1) k]  < e^{-\delta \mu / 3} = e^{-\gamma / 3} = e^{-(2\ln n + \ln \ell)} = 1/\ell n^2.
\end{align*}
Applying the union bound over all vertices $v\in V$ and all iterations $i\in [\ell]$, we obtain that with probability at least $1-1/n$ it holds that $X_{v,i}\leq k+\gamma$ for all $v$ and $i$.
%
\end{proof}

\begin{lemma}
\label{lemma:bi-test}
If $G$ is bipartite then all vertices output \accept ~in Algorithm~\ref{alg:dist-bip-test-det}. If $G$ is $\epsilon$-far from being bipartite, then with probability at least $2/3$ there exists a vertex that outputs \reject.
\end{lemma}
\begin{proof}
If $G$ is bipartite then all vertices output \accept ~in Algorithm~\ref{alg:dist-bip-test-det}, because there are no odd cycles and thus no violation detecting walks.

If $G$ is $\epsilon$-far from bipartite, we use Lemma \ref{lem:walks-detect}, in conjunction with Lemma~\ref{lemma:cong} with parameters $k=2$, $\ell=L$ and $\gamma=3(2\ln{n}+\ln{L})$ as used by Algorithm~\ref{alg:dist-bip-test-det}. By a union bound the probability to accept $G$ will be bounded by $1/4+1/n<1/3$ (assuming $n>12$), providing for the required bound on the rejection probability.
\end{proof}

Lemma~\ref{lemma:bi-test}, with the communication complexity analysis of Algorithm~\ref{alg:dist-bip-test-det}, gives Theorem~\ref{thm:bi}.

\section{Distributed test for cycle-freeness}
\label{sec:cycle}
In this section, we give a distributed algorithm to test if a graph $G$ with $m$ edges is cycle-free
or if at least $\epsilon m$ edges have to be removed to make it so. Intuitively, in order to search for cycles, one can run a breadth-first search (BFS) and have a vertex output \reject~if two different paths reach it. The downside of this exact solution is that its running time depends on the diameter of the graph. To overcome this, a basic approach would be to run a BFS from each vertex of the graph, but for shorter distances. However, running multiple BFSs simultaneously is expensive, due to the congestion on the edges. Instead, we use a simple prioritization rule that drops BFS constructions with lower priority, which makes sure that one BFS remains alive.\footnote{A more involved analysis of multiple prioritized BFS executions was used in \cite{Holzer2012}, allowing all BFS executions to fully finish in a short time without too much delay due to congestion. Since we require a much weaker guarantee, we can avoid the strong full-fledged prioritization algorithm of \cite{Holzer2012} and settle for a simple rule that keeps one BFS tree alive. Also, the multiple BFS construction of \cite{LenzenP13} does not fit our demands as it may not reach all desired vertices within the required distance, in case there are many vertices that are closer.}

Instead, our technique consists of three parts. First, we make the graph $G$ sparser, by removing each of its edges independently with probability $\epsilon/2$. We denote the sampled graph by $G'$ and prove that if $G$ is far from being cycle-free then so is $G'$, and in particular, $G'$ contains a cycle.

Then, we run a partial BFS over $G'$ from each vertex, while prioritizing by ids: each vertex keeps only the BFS that originates in the vertex with the largest id and drops the rest of the BFSs. The length of this procedure is according to a threshold $T=20\log{n}/\epsilon$. This gives detection of a cycle that is contained in a component of $G'$ with a low diameter of up to $T$, if such a cycle exists, since a surviving BFS covers the component. Such a cycle is also a cycle in $G$. If no such cycle exists in $G'$, then $G'$ has a some component with diameter larger than $T$. For large components, we take each surviving BFS that reached some vertex $v$ at a certain distance $\ell$, and from $v$ we run a new partial BFS in the \emph{original} graph $G$. These BFSs are again prioritized, this time according to the distance $\ell$. Our main tool here is proving a claim that says that with high probability, if there is a shortest path in $G'$ of length $T/2$ between two vertices, then there is a cycle in $G$ between them of length at most $T$. This allows our BFSs on $G$ to find such a cycle.

We start with the following combinatorial lemma that shows the above claim.

\begin{lemma}
  Given a graph $G$, let $G'$ be obtained by deleting each edge in $G$ with probability
  $\epsilon/2$, independently of other edges. Then, with probability at least $1-1/n^3$, every vertex $v\in
  G'$ that has a vertex $w\in G'$ at a distance $10\log n/\epsilon$, has a closed path passing
  through it in $G$, that contains a simple cycle, of length at most $20\log n/\epsilon$.
  \label{lem:main}
\end{lemma}

\begin{proof}
  First, we show that for every pair $u,v$ of vertices in $G$ that are at a distance of $\ell = 10\log
  n/\epsilon$, one of the shortest paths between $u$ and $v$ is removed in the graph $G'$ with high
  probability. For a pair of vertices $u$ and $v$ at a distance $\ell$ in $G$, the probability that a
  shortest path is not removed is $(1-\epsilon/2)^{\ell}$, which is at most $1/n^5$. Therefore, by a union
  bound over all pairs of vertices, with probability at least $1-1/n^3$, at least one edge is
  removed from at least one shortest path between every pair of vertices that are at a distance of
  $10\log n/\epsilon$. Conditioned on this, we prove the lemma.

  Now, suppose that $v$ and $w$ are two vertices in $G'$ at a distance of $10\log n/\epsilon$. Let
  $P'$ be this shortest path in $G'$. Suppose $P$ is the shortest path between $v$ and $w$ in $G$.
  If $|P|<10\log n/\epsilon$, then this path is no longer present in $G'$ (and thus distinct from
  $P'$) and $P\cup P'$ is a closed path in $G$, passing through $v$ that has a simple cycle of
  length at most $20\log n/\epsilon$. If $|P|=10\log n/\epsilon$, then there are at least two
  shortest paths between $v$ and $w$ in $G$ of length $10\log n/\epsilon$, the one in $G'$ and one
  that was removed, which we choose for $P$. Therefore, $P\cup P'$ is a closed path passing through
  $v$ of length at most $20\log n/\epsilon$, and hence contains a simple cycle of length at most
  $20\log n/\epsilon$ in it.
\end{proof}

Next, we prove that indeed there is a high probability that $G'$ contains a cycle if $G$ is far from being cycle-free.
\begin{claim}
  If $G$ is $\epsilon$-far from being cycle-free, then with probability at least $1 - e^{-\epsilon^2
  m/32}$, $G'$ is $\epsilon/4$-far from being cycle-free.
  \label{obs:distance}
\end{claim}
\begin{proof}
  The graph $G'$ is obtained from $G$ by deleting each edge with probability $\epsilon/2$
  independently of other edges. The expected number of edges that are deleted is $\epsilon m/2$.
  Therefore, by the Chernoff Bound from Fact~\ref{fact:chernoff}, the probability that at least $3\epsilon m/4$ edges are
  deleted is at most $\exp(-\epsilon^2 m/32)$, and the claim follows.
\end{proof}

We now describe a multiple-BFS algorithm that takes as input a length $t$ and a priority condition $\mathcal{P}$ over vertices, and starts performing a BFS from each vertex of the graph. This is done for $t$ steps, in each of which a vertex keeps only the BFS with the highest priority while dropping the rest. Each vertex also
maintains a list $L_v$ of BFSs that have passed through it. The list $L_v$ is a
list of $3$-tuples $(\id_u,\ell,\id_p)$, where $\id_u$ is the id of the root of the BFS, $\ell$ is the
depth of $v$ in this BFS tree and $\id_p$ is the id of the parent of $v$ in the BFS
tree. Initially, each vertex $v$ sets $L_v$ to include a BFS starting from itself, and then continues this BFS by sending
the tuple $(\id_v,1,\id_v)$ to all its neighbors, where $\id_v$ is the identifier of the vertex $v$.
In an intermediate step, each vertex $v$ may receive a BFS tuple from each of its neighbors. The vertex $v$
then adds these BFS tuples to the list $L_v$ and chooses one among $L_v$ according to the
priority condition $\mathcal{P}$, proceeding with the respective BFS and discontinuing the rest. Even when
a BFS is discontinued, the information that the BFS reached $v$ is stored in the list $L_v$.

Algorithm~\ref{alg:priority-bfs} gives a formal description of the breadth-first search that we
use in the testing algorithm for cycle-freeness.

\begin{algorithm}[htbp]
  \caption{BFS with a priority condition\label{alg:priority-bfs}}
  \KwIn{Length $L$, Priority condition $\mathcal{P}$}
    \KwVars{$L_v$ list of BFS tuples passing through $v$}
    \Simul{vertex $v$}
    {
	Initialize $L_v$ to $\{(\id_v,0,\id_v)\}$.\\
      Send $(\id_v,1,\id_v)$ to all neighbors of $v$.
    }
    \Perform{$L$ times}
    {
      \Simul{vertex $v$}
      {
	\If{$v$ receives $(\id_{u_1},\ell_1,\id_{p_1}), \ldots, (\id_{u_r},\ell_r,\id_{p_r})$ from its neighbors}
	{
	  Add $(\id_{u_1},\ell_1,\id_{p_1}), \ldots, (\id_{u_r},\ell_r,\id_{p_r})$ to $L_v$.

	  Select $(\id_{u_j},\ell_j,\id_{p_j})$ from $L_v$ according to $\mathcal{P}$ over $\id_{u_i}$

	  Send $(\id_{u_j}, \ell_j + 1, \id_v)$ to all neighbors of $v$ except $p_j$.
        }
      }
    }
\end{algorithm}
%

We now give more informal details of the test for cycle-freeness. By Lemma~\ref{lem:main}, we
know that if there is a vertex $v$ in $G'$ that has a vertex $w$ at a distance of $T/2=10\log
n/\epsilon$, then there is a closed path in $G$ starting from $v$ that contains a cycle of length
$20\log n/\epsilon$. In the first part, each vertex gets its name as its vertex id, and performs a
BFS on the graph $G'$ in the hope of finding a cycle. The BFS is performed using
Algorithm~\ref{alg:priority-bfs}, where the priority condition in the intermediate steps is
selecting the BFS with the lowest origin id. If the cycle is present in a component of diameter at
most $20\log n/\epsilon$ in $G'$, then it is discovered during this BFS. To check if there is a
cycle, one needs to find if there are appropriate tuples $(\id_{u},\ell,\id_p)$ and
$(\id_{u},\ell',\id_{p'})$ in $L_v$, for some vertex $v$.

If no cycle is discovered in this step, then we change the ids of the vertices in the following way:
The id of each vertex $v$ is now a tuple $(\ell,v)$ where $\ell$ is the largest depth at which $v$ occurs
in a BFS tree among all the breadth-first searches that reached $v$.  We perform a BFS in $G$ using
Algorithm~\ref{alg:priority-bfs}, where the priority condition is to pick the BFS whose root has the
lexicographically highest id. If there is some vertex with $\ell \geq 10\log n/\epsilon$, then the
highest priority vertex is such a vertex, and by Lemma~\ref{lem:main}, the BFS starting from that
vertex will detect a cycle in $G$.

Algorithm~\ref{alg:test-cycle-free} gives a formal description of the tester for cycle-freeness.

\begin{algorithm}[htbp]
  \caption{Cycle-freeness test\label{alg:test-cycle-free}}
  \KwVars{$L_v$ list of BFS tuples passing through $v$, vertex identifier $\id_v$}
  \tcpy{Construct $G'$ by deleting edges with probability $\epsilon/2$.}
  \Simul{vertex $v$}
  {
    For each neighbor $u<v$, mark the edge $e=(u,v)\in G$ with probability $\epsilon/2$ for deletion.

    Send each marked edge $e=(u,v)$ to its corresponding $u$.

    Set $\id_v = v$.
  }
  \Simul{vertex $v$}
  {
    Delete all edges incident on $v$ that have been marked for deletion.
  }
  \tcpy{Search for cycles in small diameter components.}
  \Use{Algorithm~\ref{alg:priority-bfs}}
  {
    perform BFS on $G'$ for $20\log n/\epsilon$ steps, with the priority condition being choosing the
    BFS with the lowest root id.
  }
  \Simul{vertex $v$}
  {
    If $L_v$ contains two tuples $(\id_{u},\ell,\id_{p})$ and $(\id_{u},\ell',\id_{p'})$, output \reject.

    Set $\id_v = (\ell_j,v)$ where $\ell_j$ is the highest among all tuples $(\id_{u_i},\ell_i,\id_{p_i})$ in
    $L_v$.
    \label{step:set-id}
  }
  \Use{Algorithm~\ref{alg:priority-bfs}}
  {
    perform BFS on $G$ for $10\log n/\epsilon$ steps, with the priority condition being choosing the BFS
    with the lexicographically highest root id.
  }
  \Simul{vertex $v\in G$}
  {
    If $L_v$ contains two tuples $(\id_{u}, \ell_j, \id_{p})$ and $(\id_{u}, \ell', \id_{p'})$, output \reject.
  }
  \Simul{vertex $v\in G$}
  {
    output \accept, if $v$ did not output \reject~ yet.
  }
\end{algorithm}

We now prove the correctness of the algorithm.

\begin{theorem}
Algorithm~\ref{alg:test-cycle-free} is a distributed $\epsilon$-test in the general graph model
for the property of being cycle-free, that requires $O(\log n/\epsilon)$ rounds.
%
  \label{thm:cycle-free-correctness}
\end{theorem}

\begin{proof}
    Notice that a vertex in Algorithm~\ref{alg:test-cycle-free} outputs \reject~only when it detects a cycle.
      Therefore, if $G$ is cycle-free, then every vertex outputs \accept~with probability $1$.

      Suppose that $G$ is $\epsilon$-far from being cycle-free. Notice that, with probability at
      least $1 - 1/n^3$, the assertion of Lemma~\ref{lem:main} holds.  Furthermore, from
      Claim~\ref{obs:distance}, we know that $G'$ is $\epsilon/4$-far from being cycle-free,
      with probability $1-e^{-\epsilon^2 m/32}$, and hence contains at least one cycle.  This cycle
      could be in a component of diameter less than $20\log n/\epsilon$ or it could be in a
      component of diameter at least $20\log n/\epsilon$ in $G'$. We analyse the two cases
      separately.

      Suppose there is a cycle in a component $C$ of $G'$ of diameter at most $20\log n/\epsilon$. Let $u$
      be the vertex with the smallest id in $C$. In Algorithm~\ref{alg:test-cycle-free}, the BFS
      starting at $u$ is always propagated at any intermediate vertex due to the priority
      condition. Furthermore, since the diameter of $C$ is at most $20\log n/\epsilon$, this BFS
      reaches all vertices of $C$.  Hence, this BFS detects the cycle and at least one vertex
      in $C$ outputs \reject.

      On the other hand, if the cycle is present in a component in $G$ of diameter at least $20\log
      n/\epsilon$, then after Step~$\ref{step:set-id}$ of the algorithm, each vertex $v$ gets the
      length of the longest path from the origin, among all the BFSs that reached $v$, as the
      first component of its id. The vertex $v$ that gets the lexicographically highest id in the
      component has a vertex $w$ that is at least $10\log n/\epsilon$ away in $G'$, since the radius
      of the component is at least half the diameter. Therefore, by Lemma~\ref{lem:main}, it is part
      of a cycle of length at most $20\log n/\epsilon$ in $G$. Hence, the vertex with the highest
      priority in the BFS on $G$ is a vertex $u$ that has a vertex at a distance of at least $10\log
      n/\epsilon$ in $G'$, and there is a walk through $u$ that contain a simple cycle of length at
      most $20\log n/\epsilon$. At least one vertex on this simple cycle will output \reject~when
      Algorithm~\ref{alg:test-cycle-free} is run on $G$.

      The number of rounds is $O(\log n/\epsilon)$ since Algorithm~\ref{alg:test-cycle-free}
      performs two breadth-first searches in the graph with this number of rounds.
\end{proof}

\section{Lower bounds for testing bipartiteness and cycle-freeness}
\label{sec:lowerbounds}
In this section, we prove that any distributed algorithm for $\epsilon$-testing bipartiteness or cycle-freeness in
bounded-degree graphs requires $\Omega(\log n)$ rounds of communication\footnote{Our lower bound applies even to the less restricted LOCAL model of communication, which does not limit the size of the messages.}. We
construct bounded-degree graphs that are $\epsilon$-far from being bipartite, such that all cycles
are of length $\Omega(\log n)$.  We argue that any distributed algorithm that runs in $O(\log n)$ rounds does not detect a witness for non-bipartiteness. We also show that the same
construction proves that every distributed algorithm for $\epsilon$-testing cycle-freeness
requires $\Omega(\log n)$ rounds of communication.
Formally, we prove the following theorem.
\begin{theorem}
\ThmLBBi
  \label{thm:bipart-lower-bound}
\end{theorem}

To prove Theorem~\ref{thm:bipart-lower-bound}, we show the existence of a graph $G'$ that is far from being bipartite, but all of its cycles are at least of logarithmic length. Since in $T$ rounds of a distributed algorithm, the output of every vertex cannot depend on vertices that are at distance greater than $T$ from it, no vertex can detect a cycle in $G'$ in less than $O(\log{n})$ rounds, which proves Theorem~\ref{thm:bipart-lower-bound}. To prove the existence of $G'$ we use the probabilistic method with alterations, and prove the following.
\begin{lemma}
  Let $G$ be a random graph on $n$ vertices where each edge is present with probability $1000/n$.
  Let $G'$ be obtained by removing all edges incident with vertices of degree greater than $2000$, and
  one edge from each cycle of length at most $\log n/\log 1000$. Then with probability at least
  $1/2 - e^{-100} - e^{-n}$, $G'$ is $1/100$-far from being bipartite.
  \label{lem:graph-far}
\end{lemma}

Since a graph that is $\epsilon$-far from being bipartite is also $\epsilon$-far from being cycle-free, we immediately obtain the same lower bound for testing cycle-freeness, as follows.

\begin{theorem}
\ThmLBCycle
  \label{cor:cycle-lower-bound}
\end{theorem}

The rest of this section is devoted to proving Lemma~\ref{lem:graph-far}. We need to show three properties of $G'$: (a) that it is far from being bipartite, (b) that it does not have small cycles, and (c) that its maximum degree is bounded. We begin with the following definition, which is similar in spirit to being far from satisfying a property and which will assist us in our proof.
\begin{definition}
 A graph $G$ is $k$-removed from being bipartite if at least $k$ edges have to be removed from $G$ to
 make it bipartite.
 \label{def:k-removed}
\end{definition}

Note that a graph $G$ with maximum degree $d$, is $\varepsilon$-far from being bipartite if it is
$\varepsilon dn$-removed from being bipartite.

Let $G$ be a random graph on $n$ vertices where for each pair of vertices, an edge is present with
probability $1000/n$. The expected number of edges in the graph is $500(n-1)$. Since the edges are
sampled independently with probability $1000/n$, by the Chernoff Bound from
Fact~\ref{fact:chernoff}, with probability at least $1-e^{-10n}$ the graph has at least $400n$
edges.  We now show that $G$ is far from being bipartite, with high probability.
\begin{lemma}[\textbf{far from being bipartite}]
  With probability at least $1-e^{-199n}$, $G$ is $20n$-far from being bipartite.
  \label{lem:farness}
\end{lemma}
\begin{proof}
  Fix a bipartition $(L,R)$ of the vertex set of $G$ such that $|L|\ge n/2$.  For each pair of
  vertices $u,v\in L$, let $X_{u,v}$ be a random variable which is $1$ if the edge $(u,v)$ is
  present in $G$ and $0$ otherwise. Its expected value is $E[X_{u,v}] = 1000/n$. The random variable
  $X=\sum_{u,v\in L} X_{u,v}$ counts the number of edges within $L$. By the linearity of expectation, $E[X] \geq \binom{n/2}{2}1000/n \geq 30n$. Since the random variables $X_{u,v}$ are
  independent, by the Chernoff Bound from Fact~\ref{fact:chernoff}, we have that $\Pr[X < 20n] \le \exp(-200n)$. Therefore,
  with probability at least $1-\exp(-200n)$, there are at least $20n$ edges within $L$. The total
  number of such bipartitions of $G$ is at most $2^{n-1}$.
  Taking a union bound over all such bipartitions, the probability that at least one of the
  bipartitions contains less than $20n$ edges within its $L$ side is at most $\exp(-199n)$, and the lemma
  follows.
\end{proof}

The expected degree of a vertex $v$ in $G$ is $1000(1-1/n)$. Therefore, by the Chernoff Bound from Fact~\ref{fact:chernoff}, the
probability that the degree of $v$ is greater than $2000$ is at most $\exp( -300(1-1/n))$.  We now
show that, with sufficiently high probability, the number of edges that are incident with high
degree vertices is small. We can remove all such edges to obtain a bounded-degree graph that is
still far from being bipartite.

\begin{lemma}[\textbf{mostly bounded degrees}]
  With probability at least $1 - e^{-100}$, there are at most $n$ edges that are incident with
  vertices of degree greater than $2000$ in $G$.
  \label{lem:edges-high-degre}
\end{lemma}
\begin{proof}
  For a pair $u,v$ of vertices, the probability that there is an edge between them and that one of
  $u$ or $v$ is of degree greater than $2000$ is $\Pr[(u,v)\in E]\cdot\Pr[u \text{ or } v \text{ has
  degree }\ge 2000|(u,v)\in E]$. This is at most $(1000/n)\cdot 2\cdot \exp( -300(1-1/n))$.
  Therefore, the expected number of edges that are incident with a vertex of degree greater
  than $2000$ is at most $1000n\cdot \exp( -300(1-1/n))$.  By Markov's inequality, the probability
  that there are at least $n$ edges that are incident with vertices of degree greater than $2000$ is at
  most $1000\cdot \exp( -300(1-1/n))$. This completes the proof of the lemma.
\end{proof}

We now bound the number of cycles of length at most $O(\log n)$ in the graph $G$.
\begin{lemma}[\textbf{few small cycles}]
  With probability at least $1/2$, there are at most $2n$ cycles of length at most $\log n/\log
  1000$ in $G$.
  \label{lem:cycles}
\end{lemma}
\begin{proof}
  For any $k$ fixed vertices, the probability that there is a cycle among the $k$ vertices is at
  most $k!(1000/n)^k$.  Therefore the expected number of cycles in $G$ of length at most $k$ is at
  most $1000^k$. For $k=\log n/\log 1000$, this means that the expected number of cycles in $G$ of
  length at most $\log n/\log 1000$ is $n$. Therefore, with probability at least $1/2$ there are at
  most $2n$ cycles of length at most $\log n/\log 1000$ in $G$.
\end{proof}


We are now ready to prove Lemma~\ref{lem:graph-far}, completing our lower bounds. Intuitively, since $G$ does not contain many high degree vertices and many small cycles, removing them to obtain $G'$ only changes the distance from being bipartite by a small term.
\begin{proof}
  With probability $1-e^{-n}$, there are at least $400n$ edges in $G$ and by Lemma~\ref{lem:farness} $G$ is $20n$-removed from
  being bipartite. By Lemma~\ref{lem:edges-high-degre}, with probability at least $1 - e^{-100}$,
  there are at most $n$ edges incident with vertices of degree greater than $2000$ and by Lemma~\ref{lem:cycles} with
  probability at least $1/2$ there are at most $2n$ cycles of length at most $\log n/\log 1000$.
  Hence, with probability at least $1/2-e^{-100} - e^{-n}$, $G'$ is a graph with degree at most
  $2000$ that is $17n$-removed from being bipartite.  Therefore, $G'$ is $1/100$-far from being
  bipartite.
\end{proof}


\section{Discussion}
\label{sec:discussion}
This paper initiates a thorough study of distributed property testing. It provides an emulation technique for the dense graph model and constructs fast
distributed algorithms for testing triangle-freeness, cycle-freeness and bipartiteness. We also present lower bounds for both bipartiteness and triangle freeness.   

This work raises many important open questions, the immediate of which is to devise fast distributed testing algorithms for additional problems. One example is testing freeness of other small subgraphs.
More ambitious goals are to handle dynamic graphs, and to find more general connections between testability in the sequential model and the distributed model.
Finally, there is fertile ground for obtaining additional lower bounds in this setting, in order to fully
understand the complexity of distributed property testing.	

\bibliographystyle{plain}
\bibliography{citations}
\end{document}